\documentclass[a4paper,conference]{IEEEtran}
\pdfoutput=1

\usepackage{pgfplots}
\pgfplotsset{compat=newest}
\usepackage{tikz}
\usetikzlibrary{arrows,matrix,positioning,patterns}

\usetikzlibrary{calc,intersections,through,backgrounds,matrix,patterns}
\usepgfplotslibrary{fillbetween}
\pgfdeclarelayer{bg}
\pgfsetlayers{bg,main}

\tikzset{
    hatch distance/.store in=\hatchdistance,
    hatch distance=100pt,
    hatch thickness/.store in=\hatchthickness,
    hatch thickness=0.3pt
}

\usepackage[utf8]{inputenc}
\usepackage[innermargin=0.545in,outermargin=0.545in,top=0.545in,bottom=1.05in]{geometry}
\usepackage[english]{babel}
\usepackage[T1]{fontenc}
\usepackage{epsfig}
\usepackage{amsmath, amssymb, amsbsy}
\usepackage{mathdots}
\usepackage{xspace}
\usepackage[noend]{algpseudocode}
\usepackage[linesnumbered,ruled,vlined,titlenumbered]{algorithm2e}
\usepackage{algorithmicx}
\usepackage{color}
\usepackage{cite}
\usepackage{booktabs}
\usepackage{verbatim}
\usepackage[colorlinks=true,citecolor=blue,linkcolor=blue,urlcolor=blue]{hyperref}
\usepackage{url}
\usepackage{lipsum}
\usepackage{enumitem}
\usepackage{colortbl}
\usepackage{amsthm}
\usepackage{dblfloatfix}
\usepackage[capitalise]{cleveref}
\usepackage{nicefrac}

\makeatletter
\newcommand\footnoteref[1]{\protected@xdef\@thefnmark{\ref{#1}}\@footnotemark}
\makeatother

\newtheorem{theorem}{Theorem}
\newtheorem{lemma}[theorem]{Lemma}
\newtheorem{corollary}[theorem]{Corollary}

\newtheorem{construction}[theorem]{Construction}

\newtheorem{definition}{Definition}

\usepackage[final,tracking=true,kerning=true,spacing=true,factor=1100,stretch=10,shrink=20]{microtype}

\newenvironment{mymatrix}{\begin{bmatrix}} {\end{bmatrix} }

\def\ve#1{{\mathchoice{\mbox{\boldmath$\displaystyle #1$}}%
              {\mbox{\boldmath$\textstyle #1$}}%
              {\mbox{\boldmath$\scriptstyle #1$}}%
              {\mbox{\boldmath$\scriptscriptstyle #1$}}}}

\newcommand{\todo}[1]{{\color{red}[#1]}}
\newcommand{\pu}[1]{{\color{brown}[pu: #1]}}
\newcommand{\jsrn}[1]{{\color{teal}[jr: #1]}}

\definecolor{darkgreen}{rgb}{0,0.5,0}

\newcommand{\Fq}{\ensuremath{\mathbb{F}_q}}
\newcommand{\Fqm}{\ensuremath{\mathbb{F}_{q^m}}}

\newcommand{\Code}{\mathcal{C}}

\renewcommand{\a}{\ve{a}}

\renewcommand{\c}{\ve{c}}

\newcommand{\ZZ}{\mathbb{Z}}
\newcommand{\RR}{\mathbb{R}}

\newcommand{\rk}{\mathrm{rk}}

\newcommand{\x}{\ve{x}}

\renewcommand{\t}{\ve{t}}
\renewcommand{\r}{\ve{r}}

\newcommand{\wtSRWITHFIELD}[1]{\mathrm{wt}_{\mathrm{SR},\ell,#1}}
\newcommand{\wtSR}{\wtSRWITHFIELD{q}}

\newcommand{\NM}[2]{\mathrm{NM}_{#1}(#2)}

\newcommand{\npr}{\eta}
\newcommand{\nmmin}{\mu}

\newcommand{\TsetWITHARGUMENTS}[1]{\mathcal{T}_{#1}}
\newcommand{\Tset}{\TsetWITHARGUMENTS{t,\ell,\nmmin}}

\newcommand{\Tsett}{\Tset}

\usepgfplotslibrary{fillbetween}

\definecolor{constructionAcolor}{rgb}{1,0.4,0}  %
\definecolor{constructionBcolor}{rgb}{0.8,0,1}	%
\definecolor{constructionCcolor}{rgb}{0,0,1}	%
\definecolor{constructionDcolor}{rgb}{0,0,0}	%
\definecolor{constructionEcolor}{rgb}{0,0.7,0}	%

\newcommand{\SparseSkewPolys}{\mathcal{R}}
\newcommand{\SkewPolys}{\Fqm[x;\sigma]}
\newcommand{\ev}{\mathrm{ev}}
\DeclareMathOperator{\supp}{supp}
\newcommand{\alphaVec}{\ve{\alpha}}
\newcommand{\betaVec}{\ve{\beta}}

\newcommand{\Fqg}{\mathbb{F}_{q^g}}
\newcommand{\Norm}[1]{\mathcal{N}_{#1}}

\newcommand{\wtRWITHFIELDSIZE}[1]{\mathrm{wt}_{\mathrm{R},#1}}
\newcommand{\wtR}{\wtRWITHFIELDSIZE{q}}
\newcommand{\BallSize}{\mathcal{B}}
\newcommand{\ListSize}{\mathcal{L}}

\newcommand{\dSRWITHFIELDSIZE}[1]{\mathrm{d}_{\mathrm{SR},\ell,#1}}
\newcommand{\dSR}{\dSRWITHFIELDSIZE{q}}

\newcommand{\an}{C}
\newcommand{\atau}{D}
\newcommand{\Gal}{\mathrm{Gal}}
\newcommand{\LRScode}{\Code_\mathsf{LRS}^{(\a,\betaVec)}[n,k]}

\title{Bounds on List Decoding of \\ Linearized Reed--Solomon Codes}

\author{\IEEEauthorblockN{Sven Puchinger, Johan Rosenkilde}
\IEEEauthorblockA{
Department of Applied Mathematics and Computer Science, \\ Technical University of Denmark (DTU), Lyngby, Denmark\\
Email: svepu@dtu.dk, jsrn@dtu.dk
\thanks{This work has been supported by the European Union's Horizon 2020 research and innovation programme under the Marie Sklodowska-Curie grant agreement no.~713683.}
}
}

\IEEEoverridecommandlockouts

\begin{document}

\maketitle

\begin{abstract}
Linearized Reed--Solomon (LRS) codes are sum-rank metric codes that fulfill the Singleton bound with equality.
In the two extreme cases of the sum-rank metric, they coincide with Reed--Solomon codes (Hamming metric) and Gabidulin codes (rank metric).
List decoding in these extreme cases is well-studied, and the two code classes behave very differently in terms of list size, but nothing is known for the general case.
In this paper, we derive a lower bound on the list size for LRS codes, which is, for a large class of LRS codes, exponential directly above the Johnson radius.
Furthermore, we show that some families of linearized Reed--Solomon codes with constant numbers of blocks cannot be list decoded beyond the unique decoding radius.
\end{abstract}


\section{Introduction}

The sum-rank metrics is a fairly recent family of metrics: two vectors of length $n$ are split into $\ell$ blocks each and their sum-rank distance is the sum of the rank distances of the block pairs. For $\ell=n$, this is simply the Hamming metric, while for $\ell=1$ it is the rank metric.
This family was introduced in 2010 \cite{nobrega2010multishot}, motivated by applications in multi-shot network coding.
Other applications are distributed storage \cite{martinez2019universal}, other aspects of network coding \cite{martinez2019reliable}, and space-time codes \cite{shehadeh2020rate}.
There are several code constructions and decoding algorithms for codes in the sum-rank metric \cite{wachter2011partial,wachter2012rank,wachter2015convolutional,napp2017mrd,napp2018faster,martinez2018skew,boucher2019algorithm,martinez2019reliable,caruso2019residues,bartz2020fast,martinezpenas2020sumrank,byrne2020fundamental}.

The class of linearized Reed--Solomon (LRS) codes were introduced in \cite{martinez2018skew}, and for a given sum-rank metric, there is a large family of LRS codes meeting the Singleton bound for that metric.
When the sum-rank metric is actually the Hamming metric, the corresponding family of LRS codes is the well-known Reed--Solomon (RS) codes \cite{reed1960polynomial}, and when it is the rank metric, the corresponding LRS codes are the Gabidulin codes \cite{Delsarte_1978,Gabidulin_TheoryOfCodes_1985,Roth_RankCodes_1991}.
Also the restriction on the number of blocks $\ell$ and the block size $\tfrac{n}{\ell}$ of an LRS code are a mix of the two extreme cases: If the code is defined over $\Fqm$, where the rank is taken w.r.t.\ the subfield $q$, then LRS codes require $\ell <q$ and $\tfrac{n}{\ell}\leq m$.

The topic of this paper is the list size \cite{elias1957list,wozencraft1958list} of LRS codes, which is the maximum number of codewords in a ball of given radius in the sum-rank metric (maximized over all possible centers of the ball).
For the extreme cases of LRS codes, the combinatorial list decoding problem is well-studied:

RS codes in the Hamming metric ($\ell=n$) have a polynomial list size up to the relative decoding radius $J_R = 1 - \sqrt R$, where $R$ the rate, known as the Johnson radius \cite{johnson_new_1962}.
Understanding the list size beyond $J_R$ is a long-standing open problem.
We know that ``most'' RS codes allow list decoding beyond $J_R$ \cite{rudra2014every}.
On the other hand, \cite{justesen2001bounds} and \cite{ben2009subspace} showed that the list size is exponential below the channel capacity for RS codes whose dimension grows exponentially slower than the length.

For Gabidulin codes in the rank metric ($\ell=1$), Wachter-Zeh \cite{wachter2013bounds} and Raviv--Wachter-Zeh \cite{raviv2016some} adapted the arguments of \cite{justesen2001bounds, ben2009subspace} to show that \emph{all} Gabidulin codes with $n=m$ has an exponential list-size starting from $J_R$, and that \emph{some} have exponential list size already from half the minimum distance.
For the latter result, they gave explicit constructions with fixed rates $\geq 0.2$.

\subsection{Contributions}

We basically extend the known list size results for Gabidulin codes in the rank metric to almost all LRS codes in the sum-rank metric.

More specifically, in Section~\ref{sec:exponential_from_johnson_radius}, we show that \emph{all} LRS codes
have exponential list size in $n$ above a specific radius.
For LRS codes with $\ell \in o(n)$ and of smallest possible field extension degree $m = \tfrac{n}{\ell}$, this radius is the Johnson radius $J_R$.
For certain families of LRS codes with $\ell \in \Theta(n)$, we also obtain a fixed relative decoding radius slightly beyond $J_R$.
For RS codes $\ell = n$, our bound is the same as in \cite{ben2009subspace}.

In Section~\ref{sec:some_not_list-decodable}, we show that \emph{some} families of LRS codes have exponential list size directly above half the minimum distance.
For $\ell=1,2,\dots$, the constructed families may have rates $\geq 0.2,0.33,0.33,0.5,0.2,0.33,\dots$, and asymptotically $\geq 1-\Theta(\nicefrac{1}{\sqrt{\ell}})$.
This result extends \cite{raviv2016some} to LRS codes with a constant number of blocks $\ell$.

The results indicate that LRS codes in the studied parameter ranges behave
similarly to Gabidulin codes ($\ell=1$) in terms of list decodability.
\cref{fig:results_illustration} illustrates some of the results for LRS codes with $m = \tfrac{n}{\ell}$. %

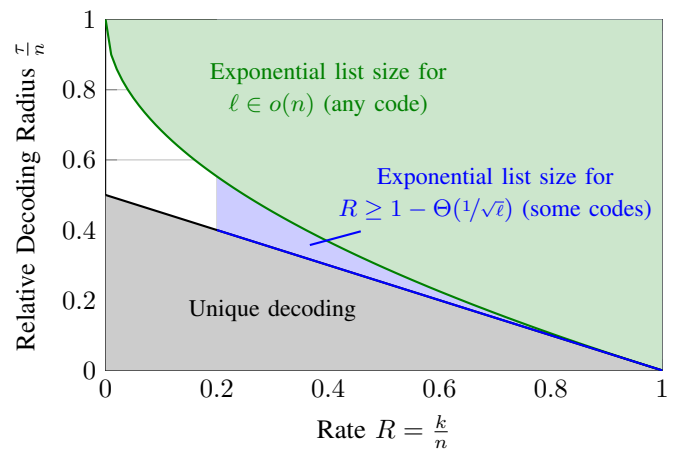
\begin{figure}[ht!]
\begin{center}
\begin{tikzpicture}
\pgfplotsset{compat = 1.3}
\begin{axis}[
	legend style={nodes={scale=0.7, transform shape}},
	width = \columnwidth,
	height = 0.7\columnwidth,
	xlabel = {{Rate $R = \tfrac{k}{n}$}},
	ylabel = {{Relative Decoding Radius $\tfrac{\tau}{n}$}},
	xmin = 0.0,
	xmax = 1.0,
	ymin = 0.0,
	ymax = 1.0,
	grid=both]

\addplot [name path=johnson_radius, domain=0:1, samples=101,unbounded coords=jump, color=darkgreen, thick] {1-x^(0.5)};
\addplot [name path=unique_radius, domain=0:1, samples=101,unbounded coords=jump, color=black, thick] {0.5*(1-x)};
\addplot [name path=all_one, domain=0:1, samples=101,unbounded coords=jump, draw=none] {1};
\addplot [name path=all_zero, domain=0:1, samples=101,unbounded coords=jump, draw=none] {0};
\addplot[black!20, forget plot] fill between[of=unique_radius and all_zero];

\addplot [name path=all_one_half, domain=0.2:1, samples=101,unbounded coords=jump, draw=none] {1};
\addplot [name path=unique_radius_half, domain=0.2:1, samples=101,unbounded coords=jump, color=blue, thick] {0.5*(1-x)};
\addplot[blue!20, forget plot] fill between[of=all_one_half and unique_radius_half];

\addplot[darkgreen!20, forget plot] fill between[of=johnson_radius and all_one];

\node[black] at (axis cs:0.3,0.17) {\small Unique decoding};
\node[darkgreen, align=center] at (axis cs:0.4,0.8) {\small Exponential list size for \\ \small $\ell \in o(n)$ (any code)};
\node (pointer_some) at (axis cs:0.35,0.35) {};
\node[blue, align=center] (label_some) at (axis cs:0.7,0.5) {\small Exponential list size for \\ \small $R \geq 1-\Theta(\nicefrac{1}{\sqrt{\ell}})$ (some codes)};
\draw[blue, thick] (label_some) to (pointer_some);

\end{axis}
\end{tikzpicture}
\end{center}
\vspace{-0.5cm}
\caption{Illustration of the results for LRS codes with block size $\tfrac{n}{\ell}=m$, where $m$ is the extension degree of the underlying field extension.}
\label{fig:results_illustration}
\end{figure}

\section{Preliminaries}

Let $q$ be a prime power and $m$ be a positive integer.
We denote by $\Fq$ the finite field of size $q$ and by $\Fqm$ its extension field of extension degree $m$. We will make use of the fact that $\Fqm$ is an $m$-dimensional $\Fq$-vector space, which means that the rank $\rk_{\Fq}(\x) := \dim_{\Fq} \langle x_1,\dots,x_{\npr}\rangle_{\Fq}$ of a vector $\x \in \Fqm^{\npr}$ is well-defined.

\subsection{Sum-Rank Metric}

Throughout this paper, let $\ell,\eta,n \in \ZZ_{> 0}$ such that $n=\ell \eta$.
We say that $n$ is the \emph{code length}, $\ell$ is the \emph{number of blocks}, and $\eta$ is the \emph{block size}.
Such a triple together with a finite field $\Fqm$ induces a sum-rank metric defined as follows.
\begin{definition}
The ($\ell$-)sum-rank weight on $\Fqm^n$ is defined as
\begin{equation*}
\wtSR \, : \, \Fqm^n \to \ZZ_{\geq 0}, \quad \x  \mapsto \textstyle\sum_{i=1}^{\ell} \rk_{\Fq}(\x_i),
\end{equation*}
where we write $\x = \big[ \x_1 | \x_2 | \dots | \x_\ell \big]$ with $\x_i \in \Fqm^{\npr}$.
Furthermore, the ($\ell$-)sum-rank distance is defined as
\begin{align*}
\dSR \, : \, \Fqm^n \times \Fqm^n \to \ZZ_{\geq 0}, \quad [\x_1,\x_2] \mapsto \wtSR(\x_1-\x_2).
\end{align*}
\end{definition}
For $\ell=1$, the sum-rank metric coincides with the rank metric and for $\ell=n$, it is the Hamming metric.

The ball with center $\r \in \Fqm^n$ and radius $\tau$ w.r.t.\ the $\ell$-sum-rank metric is defined as
\begin{align*}
\BallSize_{\tau,\ell}(\r) := \left\{ \x \in \Fqm^n \, : \, \dSR(\x,\r) \leq \tau \right\}.
\end{align*}
If $\Code \subseteq \Fqm^n$ is a code (i.e.~any subset of $\Fqm^n$), and $\tau$ is any positive integer at most $n$, then the $\tau$-\emph{list size} of $\Code$ (wrt.~the $\ell$-sum-rank) is defined as:
\[
  \ListSize(\Code, \tau) := \max_{\r \in \Fqm} |\Code \cap \BallSize_{\tau,\ell}(\r)| \ .
\]

\subsection{Conjugacy in a Finite Field}

Let $\psi \in \Gal(\Fqm/\Fq)$ be an element of the Galois group of the field extension $\Fqm/\Fq$.
This means that $\psi = \phi_q^s$ is a power of the Frobenius automorphism $\phi_q \, : \, \Fqm \to \Fqm, \, a \mapsto a^q$.
The fixed field $\Fqm^\psi$ of $\psi$ is the set $\{a \in \Fqm \, : \, \psi(a)=a\}$.
Note that $\Fqm^\psi$ is indeed a field and $\Fq \subseteq \Fqm^\psi \subseteq \Fqm$.
Further $\Fqm^\psi = \Fq$ exactly when $\psi$ generates the group $\Gal(\Fqm/\Fq)$, which is exactly when $\psi = \phi_q^s$ for an $s$ satisfying $\gcd(m, s) = 1$.
The \emph{norm} map w.r.t.\ $\psi$ is the multiplicative homomorphism
\begin{align*}
\Norm{\psi} \, : \, \Fqm^* &\to (\Fqm^\psi)^* , \\
a &\mapsto \prod_{i=0}^{[\Fqm:\Fqm^\psi]-1} \psi^i(a).
\end{align*}
We define conjugacy as follows:
\begin{definition}[\!\!\cite{lam1985general,lam1988vandermonde}]
Two elements $a,b \in \Fqm^\ast$ are \emph{conjugates w.r.t.\ $\psi$} if they have the same norm $\Norm{\psi}(a) = \Norm{\psi}(b)$.
\end{definition}
This defines an equivalence relation on $\Fqm^\ast$ with $|\Fqm^\psi|-1$ conjugacy classes of size $\tfrac{q^m-1}{|\Fqm^\psi|-1}$ each.
By Hilbert's Theorem 90, then $a$ and $b$ are conjugates if and only if there is a $c \in \Fqm$ such that $a = b \psi(c)c^{-1}$.

The norm, and hence also conjugacy, turns out not to depend on $\psi$ but only on $\Fqm^\psi$.
Throughout this paper, $\sigma \in \Gal(\Fqm/\Fq)$ will be chosen to have fixed field $\Fq$, in which case $\Norm{\sigma}(a) = \Norm{\phi_q}(a) = a^{1 + q + \ldots + q^{m-1}}$, and conjugacy wrt.~$\sigma$ is the commonly used notion of conjugacy for the extension $\Fqm/\Fq$.

\subsection{Skew Polynomials}

Skew polynomials were first introduced by Ore in \cite{ore1933theory} in a quite general setting. In this paper, we use the following special case (in particular, we do not use derivations):
\begin{definition}
The ring of skew polynomials $\SkewPolys$ is a set of formal polynomials
\begin{align*}
\left\{ f = \sum_{i=0}^{d} f_i x^i \, : \, f_i \in \Fqm, \, d \in \ZZ_{\geq 0} \right\}
\end{align*}
equipped with ordinary (component-wise) addition
\begin{align*}
f+g = \sum_{i\geq 0} (f_i+g_i) x^i
\end{align*}
and the multiplication rule
\begin{align*}
x \cdot a = \sigma(a) \cdot x,
\end{align*}
extended to polynomials by associativity and distributivity.
\end{definition}

Multiplication in $\SkewPolys$ is non-commutative whenever $\sigma \neq \mathsf{id}$.
The \emph{degree} of a skew polynomial is defined as
\begin{align*}
\deg f := \begin{cases}
\max\{i \, : \, f_i \neq 0\}, &\text{if } f\neq0, \\
-\infty, &\text{otherwise.}
\end{cases}
\end{align*}
For an integer $n$, we denote by $\SkewPolys_{<n}$ the set of skew polynomials of degree smaller than $n$.

\subsection{Generalized Operator Evaluation}

The codes studied in this paper are evaluation codes of skew polynomials, where the used evaluation map is the \emph{generalized operator evaluation} \cite{leroy1995pseudolinear}, defined as follows.
Let $f \in \SkewPolys$ and $a \in \Fqm^\ast$. Define
\begin{align*}
f(\cdot)_a \, : \, \Fqm &\to \Fqm, \\
\beta &\mapsto \sum_{i\geq 0} f_i \sigma^i(\beta) N_i(a),
\end{align*}
where $N_i(a) := \prod_{j=0}^{i-1} \sigma^j(a)$.
Due to the $\Fq$-linearity of $\sigma$, the map $f(\cdot)_a$ is $\Fq$-linear for a fixed $a$.

For $\a := [a_1,\dots,a_\ell] \in \Fqm^\ell$ and $\betaVec := [\beta_1,\dots,\beta_\npr] \in \Fqm^\npr$, define the multi-point evaluation map
\begin{align*}
&\ev_{\a,\betaVec}(\cdot) \, : \, \SkewPolys \to \Fqm^n \\
&f \mapsto \Big[ f(\beta_1)_{a_1}, \dots, f(\beta_\npr)_{a_1},  f(\beta_1)_{a_2}, \dots, f(\beta_\npr)_{a_\ell}\Big].
\end{align*}

We will use such an evaluation map only when $\a$ and $\betaVec$ satisfy certain criteria, which we give a name:

\begin{definition}
  A pair $(\a, \betaVec) \in \Fqm^\ell \times \Fqm^\npr$ is said to be an \emph{evaluation pair} (wrt.~$\sigma$) if the elements of $\a$ are in distinct conjugacy classes, and the elements of $\betaVec$ are linearly independent over $\Fq$.
\end{definition}

The following is well-known.

\begin{lemma}[{Collection of results in \cite{martinez2018skew}, or \cite[Proposition~1.3.7]{caruso2019residues}}]\label{lem:properties_multi_point_evaluation}
  Let $(\a, \betaVec) \in \Fqm^\ell \times \Fqm^\npr$ be an evaluation pair.
  Then,
\begin{itemize}
\item the restricted mapping $\ev_{\a,\betaVec}(\cdot)|_{\SkewPolys_{<n}}$ is bijective and
\item we have for any non-zero $f \in \SkewPolys_{<n}$
\begin{align*}
\wtSR(\ev_{\a,\betaVec}(f)) \geq n-\deg f.
\end{align*}
\end{itemize}
\end{lemma}

\subsection{Linearized Reed--Solomon Codes}

\begin{definition}[\!\!\cite{martinez2018skew}]
  Consider $\ell, \eta, n$ which induces a sum rank metric on $\Fqm^n$ such that $\eta \leq m$ and $\ell < q$, and $\sigma$ an automorphism of $\Fqm$ with fixed field $\Fq$.
  Let $1 \leq k \leq n$ be a dimension and fix an evaluation pair $(\a, \betaVec) \in \Fqm^\ell \times \Fqm^\npr$.
  The corresponding linearized Reed--Solomon (LRS) code is defined by
  \begin{align*}
    \LRScode := \left\{ \ev_{\a,\betaVec}(f) \, : \, f \in \SkewPolys_{<k} \right\}.
  \end{align*}
\end{definition}

It is immediately clear from \cref{lem:properties_multi_point_evaluation} that the minimum $\ell$-sum-rank distance of $\LRScode$ is $d=n-k+1$, which is the maximal possible due to a Singleton-analogue for the $\ell$-sum-rank distance \cite[Proposition 34]{martinez2018skew}.
For $\ell=1$, the codes coincide with Gabidulin codes \cite{Delsarte_1978,Gabidulin_TheoryOfCodes_1985,Roth_RankCodes_1991} and for $\ell=n$, they are generalized Reed--Solomon codes \cite{reed1960polynomial}.

\subsection{Number of Matrices of a Given Rank}

The number of matrices in $\Fq^{m \times \eta}$ of rank $t \leq \max\{\eta,m\}$ is given by \cite{migler2004weight}
\begin{align*}
\NM{q}{m,\npr,t} = \prod_{i=0}^{t-1} \frac{\left(q^m-q^i\right)\left(q^\eta-q^i\right)}{q^t-q^i},
\end{align*}
and we can bound it from below and above by \cite{ott2021bounds,migler2004weight}
\begin{align*}
\gamma_q^{-1} q^{t(m+\eta-t)} \leq \NM{q}{m,\npr,t} \leq \gamma_q q^{t(m+\eta-t)},
\end{align*}
where $\gamma_q := \prod_{i=1}^{\infty} (1-q^{-i})^{-1}$ is a constant depending only on $q$, which is monotonically decreasing in $q$ with a limit of $1$, and e.g.~$\gamma_2 \approx 3.463$, $\gamma_3 \approx 1.785$, and $\gamma_4 \approx 1.452$.

The number of vectors in $\Fqm^n$ of $\ell$-sum-rank weight $t$ is defined as $\mathcal{N}_{q,\npr,m}(t,\ell)$, and we have
\begin{align*}
\mathcal{N}_{q,\npr,m}(t,\ell) = \sum_{\t \in \Tsett} \prod_{i=1}^{\ell} \NM{q}{m,\npr,t_i}.
\end{align*}
We make use of the following lower bound on $\mathcal{N}_{q,\npr,m}(t,\ell)$. %

\begin{lemma}[\!\!\cite{ott2021bounds}]\label{lem:lower_bound_sphere_size}
We have
\begin{align*}
\mathcal{N}_{q,\npr,m}(t,\ell) \geq \begin{cases}
q^{t(\eta+m-\frac{t}{\ell})} \gamma_q^{-\ell}, &\ell \mid t, \\
q^{t(\eta+m-\frac{t}{\ell})-\frac{\ell}{4}} \gamma_q^{-\ell}, &\ell \nmid t.
\end{cases}
\end{align*}
\end{lemma}

\vspace{0.3cm}
\section{LRS Codes with Exponential List Size Above the (Hamming) Johnson Radius}\label{sec:exponential_from_johnson_radius}

Our list-size lower bounds are based on the following simple counting strategy, originally used in \cite{justesen2001bounds} for the Hamming metric and expounded in \cite{ben2009subspace}:

\begin{lemma}
  \label{lem:listsize_lemma}
  Consider $\ell, \eta, n$ which induces a sum rank metric on $\Fqm^n$ such that $\eta \leq m$ and $\ell < q$, and $\sigma$ an automorphism of $\Fqm$ with fixed field $\Fq$, and $\Code = \LRScode$ a linearized RS code.
  Let $\tau$ be a positive integer less than $d = n-k+1$.
  Let $S \subseteq \SkewPolys_{<n}$ be a set of polynomials such that each element's $\ev_{\a,\betaVec}$-image has $\ell$-sum-rank weight at most $\tau$.
  Let $s \leq n-k$ be the number of monomials among $x^k, x^{k+1},\ldots,x^{n-1}$ such that there is at least one element of $S$ with non-zero coefficient for that monomial.
  Then
  \[
    \ListSize(\Code, \tau) \geq \frac {|S|}{q^{ms}} \ .
  \]
\end{lemma}
\begin{proof}
  Divide $S$ into disjoint subsets such that all polynomials in each subset have the same coefficients of degree $k, k+1, \ldots, n-1$.
  Note that there are at most $q^{ms}$ such subsets.
  Let $S' \subset S$ be such a subset of maximal cardinality.
  By the Pigeonhole principle, then $L := |S'| \geq |S| q^{-ms}$.
  Write $S' = \{ f_1, \ldots, f_L \}$, and set
  \begin{align*}
    \r &:= \ev_{\a,\betaVec}(f_1), \\
    \c_i &:= \ev_{\a,\betaVec}(f_i-f_1) \quad \forall \, i=1,\dots,L.
  \end{align*}
  Then,
  \begin{align*}
    \dSR(\r,\c_i) &= \wtSR(\ev_{\a,\betaVec}(f_i)) \leq \tau
  \end{align*}
  by definition of $S$ and $\deg (f_i-f_1) <k$ since the $n-k$ top-most coefficients of the $f_i$ are the same.
  Hence,
  \begin{align*}
    \c_1,\dots,\c_L \in \Code \cap \BallSize_{\tau,\ell}(\r) \ ,
  \end{align*}
  which proves the claim.
\end{proof}

The game is now to construct sets $S$ such that the trade-off between $|S|$ and $s$ makes the value $|S|q^{-ms}$ as large as possible.
Our first bound is achieved by choosing $S$ to be all possible satisfactory polynomials:

\begin{theorem}\label{thm:list_size_ell=1}
Let $\ell,m,n,\eta,k,d$ be valid parameters of a linearized RS code $\Code$ and $\tau <d$.
Then
  \[
    \ListSize(\Code, \tau) \geq
      q^{m+\tau(m + \eta) - \frac{\tau^2}{\ell} - md} U^{-1} \ , \label{eq:number_of_codewords_in_ball_g=1}
  \]
  where $U = \gamma_{q}^{\ell}$ if $\ell \mid \tau$ and $U = (q^{1/4}\gamma_q)^{\ell}$ otherwise.
\end{theorem}
\begin{proof}
Consider the set $S$ of \emph{all} linearized polynomials of degree less than $n$ such that their $\ev_{\a,\betaVec}$-image has $\ell$-sum-rank weight at most $\tau$.
Since $\ev_{\a,\betaVec} \, : \, \SkewPolys_{<n} \to \Fqm^n$ is a bijection, we have
\begin{align*}
|S| = \sum_{t=0}^{\tau} \mathcal{N}_{q,\npr,m}(t,\ell) \geq q^{\tau(\eta+m-\frac{\tau}{\ell})} U^{-1}\ ,
\end{align*}
where the lower bound holds due to \cref{lem:lower_bound_sphere_size}.
Apply \cref{lem:listsize_lemma} with $s = n-k$.
\end{proof}

There are many ways this can be phrased as an asymptotic bound.
We give three examples: the first very general, while the second two are more specific but have clear take-aways.

\begin{corollary}
  \label{cor:exponential_general}
  Fix a prime power $q$ and some real number $\varepsilon > 0$, and an infinite subset $N \subset \mathbb{N}$.
  Let $\{ \Code_n \}_{n \in N}$ be a family of linearized RS codes such that $\Code_n$ has length $n$ over an extension field of $\Fq$.
  For each such $n \in N$, choose a decoding radius $\tau$ greater or equal to
  \[
	\tfrac{\ell m + n}{2} - \sqrt{\left(\tfrac{\ell m + n}{2}\right)^2-\ell^2\left(\tfrac{1}{4}+\log_q(\gamma_q)\right) -\ell m(d-1) - \varepsilon n)} \ ,
  \]
  where the code parameters correspond to the respective code $\Code_n$.
  Then $\ListSize(\Code_n, \tau_n) \in \Omega(q^{\varepsilon n/\ell})$.
\end{corollary}

\begin{corollary}
  \label{cor:exponential_johnson}
  Fix a rate $R$.
  In the context of \cref{cor:exponential_general}, assume that $\Code_n$ has $k = \lfloor Rn \rfloor$ and $m = \eta = n/\ell$, and where $\ell \in o(n)$ holds for the family.
  Then there is a sequence of decoding radii $\tau_n$ such that $\tau_n/n \rightarrow 1 - \sqrt{R} + \varepsilon$ and $\ListSize(\Code_n, \tau_n) \in \omega(q^{\varepsilon n})$.
\end{corollary}
\begin{proof}
  Since $\ell \in o(n)$, we may choose as $\tau_n$ the least integer greater than $n(1 - \sqrt{R} + \varepsilon)$ which is divisible by $\ell$.
  Then we may set $U = \gamma_g^\ell$ in \cref{thm:list_size_ell=1}, and observe $\log_q(\ListSize(\Code_n, \tau_n))$:
  \begin{align*}
    &\log_q(\ListSize(\Code_n, \tau_n)) \\
    \geq & \big( \tfrac 1 \ell ( n + 2\tau n - \tau^2 - nd) - \ell \log_q(\gamma_q) \\
    = & \tfrac {n^2} \ell ( 2(1 - \sqrt{R} + \varepsilon)- (1 - \sqrt{R} + \varepsilon)^2  \\
&\quad \quad \quad - (1 - \sqrt{R})) - \ell \log_q(\gamma_q) \\
    = &\tfrac {\varepsilon n^2} \ell ( 2\sqrt{R} - \varepsilon) - \ell \log_q(\gamma_q)  \\
    = &\varepsilon n \left(\tfrac n \ell (2\sqrt{R} - \varepsilon) - \tfrac \ell {\varepsilon n} \log_q(\gamma_q) \right) 
  \end{align*}
  Since $\ell \in o(n)$ then $\frac n \ell \rightarrow \infty$ and $\frac \ell {\varepsilon n} \rightarrow 0$, and hence  $\ListSize(\Code_n, \tau_n) \in \omega(q^{\varepsilon n})$.
\end{proof}
Recall that $m = \eta$ is the minimal possible field extension for a linearized RS code.
Then \cref{cor:exponential_johnson} implies that as long as $\ell$ does not grow as fast as $n$, the list size grows exponential above the Johnson radius.

The following corollary shows that slightly beyond the Johnson bound, we may even conclude an exponential list size bound for most families of LRS codes where $\ell$ grows linearly in $n$ (i.e.~very Hamming-like codes):

\begin{corollary}
  \label{cor:exponential_johnson_linear_blocks}
  Fix a rate $R$ and a constant $a \in ]0;1[$ such that $\zeta := a^2(\nicefrac 1 4 + \log_q(\gamma_q)) < R$.
  In the context of \cref{cor:exponential_general}, assume that $\Code_n$ has $k = \lfloor Rn \rfloor$ and $m = \eta = n/\ell$, and where $\ell < an$, and choose $\varepsilon > \sqrt{R} - \sqrt{R - \zeta} \in \RR_{+}$.
  Then there is a sequence of decoding radii $\tau_n$ such that $\tau_n/n \rightarrow 1 - \sqrt{R} + \varepsilon$ and $\ListSize(\Code_n, \tau_n) \in \omega(b^n)$, for some real number $b > 1$.
\end{corollary}
\begin{proof}
  We set $U = (q^{1/4}\gamma_g)^\ell = q^{\ell \zeta/a^2}$ in \cref{thm:list_size_ell=1} and observe $\log_q(\ListSize(\Code_n, \tau_n)):$
  \begin{align*}
    &\log_q(\ListSize(\Code_n, \tau_n)) \\
    \geq &\tfrac 1 \ell ( n + 2\tau n - \tau^2 - nd) - \tfrac {\ell \zeta} {a^2} \\
    \geq & \tfrac {\varepsilon n} a ( 2\sqrt{R} - \varepsilon) - \tfrac {n \zeta} a \\
    = & n \tfrac \delta a \ ,
  \end{align*}
  where $\delta := -\varepsilon^2 + 2\sqrt{R}\varepsilon - \zeta$.
  Then $\delta > 0$ is assured by the choice of $\varepsilon$.
  Choosing $b = q^{\delta / a}$ completes the proof.
\end{proof}

\section{Some LRS Codes Cannot be List Decoded at Any Radius}\label{sec:some_not_list-decodable}

In the following, if $f \in \SkewPolys$, then $\supp(f)$ denotes the set of exponents for which $f$ has a non-zero monomial.
E.g.~if $f = x^3 + 3$, then $\supp(f) = \{ 0, 3 \}$.

\subsection{Sparse Skew Polynomials}

\begin{definition}\label{def:sparse_polynomial_set}
  Consider $\ell, \eta, n$ which induces a sum rank metric on $\Fqm^n$ such that $\eta \leq m$ and $\ell < q$, and $\sigma$ an automorphism of $\Fqm$ with fixed field $\Fq$, and let $(\a, \betaVec) \in \Fqm^\ell \times \Fqm^\npr$ be an evaluation pair.
  Let $0 \leq \tau \leq n$ be a decoding radius and let further $g \in \ZZ_{> 0}$ be a sparsity index.
  We define the set
  \begin{align*}
    \SparseSkewPolys_g^{(\a,\betaVec)}(\sigma, \tau) :=
    \Big\{ f \in \SkewPolys_{<n}  \, : \, \\
    \supp(f) \subseteq g \ZZ, \,\wtSR\big(\ev_{\a,\betaVec}\big) \leq \tau \Big\}
  \end{align*}
\end{definition}
Note that in the definition of $\SparseSkewPolys_g^{(\a,\betaVec)}(\sigma, \tau)$, then $\ell, \eta, n$ are indirectly specified by $\a$ and $\betaVec$, and $q^m$ and $q$ through the domain resp.~fixed field of $\sigma$.

In the previous section, we already used \cref{lem:lower_bound_sphere_size} to show:
\begin{align*}
  \Big|\SparseSkewPolys_1^{(\a,\betaVec)}(\sigma, \tau)\Big| &= \sum_{i=0}^{\tau} \mathcal{N}_{q,\npr,m}(t,\ell) \ ,
\end{align*}
and if $\ell \mid \tau$, then we can lower-bound this value by
\begin{align}
  \Big|\SparseSkewPolys_1^{(\a,\betaVec)}(\sigma,\tau)\Big| &\geq q^{\tau\left(m + \npr -\frac{\tau}{\ell}\right)} \gamma_q^{-\ell}. \label{eq:cardinality_dense_polynomials}
\end{align}

In the following, we will use $g$ to be a divisor of $m$, and consider the automorphism $\sigma^g$, which we remark has fixed field $\Fqg$.

\begin{theorem}\label{thm:lower_bound_number_of_sparse_polys}
Let $\ell,\eta,n,m,\tau,\a,\betaVec$ be chosen as in \cref{def:sparse_polynomial_set} with the additional restrictions $\ell \mid \tau$, $g \mid m$, $\eta = m$, $g \mid \tau$, and $\betaVec$ of the form
\begin{align*}
\betaVec := [\alpha_1\gamma_1,\dots,\alpha_1\gamma_g,\alpha_2\gamma_1,\dots,\alpha_{\npr/g}\gamma_g] \in \Fqm^\npr,
\end{align*}
Then, the $\beta_i$ are linearly independent over $\Fq$ and we have
\begin{align*}
&\Big|\SparseSkewPolys_g^{(\a,\betaVec)}(\sigma,\tau)\Big| \geq q^{\frac{\tau}{g} \left(m + \npr - \frac{\tau}{\ell}\right)} \gamma_{q^g}^{-\ell}.
\end{align*}
\end{theorem}

\begin{proof}
Let $f \in \Fqm[y; \sigma^g]$ have degree $<n/g$ and define $\tilde{f} \in \SkewPolys_{<n}$ as
\begin{align*}
\tilde{f} := \sum_{i=0}^{n/g-1} f_i x^{gi}.
\end{align*}
Note that for any $a \in \Fqm^\ast$, the evaluation map $f(\cdot)_a$ is $\Fqg$-linear and the evaluation map $\tilde{f}(\cdot)_a$ is only $\Fq$-linear in general.

We need to show that we have
\begin{align*}
\wtSR\!\left(\ev_{\a,\betaVec}\!\left(\tilde{f}\right)\right) = g \cdot \wtSRWITHFIELD{q^g}\!\left(\ev_{\a',\alphaVec}(f)\right)
\end{align*}
for any $f \in \SparseSkewPolys_1^{(\a',\alphaVec)}(\sigma^g,\tfrac{\tau}{g})$,
where $\a' \in \Fqm^\ell$ is defined by
\begin{align*}
a_i' := \prod_{j=0}^{g-1}\sigma^j(a_i),
\end{align*}
and that the $a_i'$ belong to distinct conjugacy classes of $\Fqm$ w.r.t.\ $\Fqg$.
The latter claim is obvious by definition of conjugacy since
\begin{align*}
\Norm{\sigma^g}(a_i') = \Norm{\sigma}(a_i)
\end{align*}
for all $i$.

We show that for any $i=1,\dots,\ell$, we have %
\begin{align*}
&\wtR \big[ \tilde{f}\!\left(\beta_1\right)_{a_i}, \dots, \tilde{f}\!\left(\beta_\npr\right)_{a_i} \big] \\
&= \rk_{\Fq} \left\langle \tilde{f}\!\left(\alpha_1 \gamma_1\right)_{a_i}, \dots, \tilde{f}\!\left(\alpha_1 \gamma_g\right)_{a_i},\tilde{f}\!\left(\alpha_1 \gamma_1\right)_{a_i}, \dots \right. \\
&\qquad \qquad \qquad \qquad \qquad \qquad \qquad \dots, \left. \tilde{f}\!\left(\alpha_{\npr/g} \gamma_g\right)_{a_i} \right\rangle_{\Fq} \\
&\overset{\mathrm{(i)}}{=} \rk_{\Fq} \textstyle\sum_{j=1}^{\npr/g} \left\langle f\!\left(\alpha_j \gamma_1\right)_{a_i'}, \dots, f\!\left(\alpha_j \gamma_g\right)_{a_i'} \right\rangle_{\Fq} \\
&\overset{\mathrm{(ii)}}{=} \rk_{\Fq} \textstyle\sum_{j=1}^{\npr/g} \left\langle \gamma_1 f\!\left(\alpha_j \right)_{a_i'}, \dots, \gamma_g f\!\left(\alpha_j\right)_{a_i'} \right\rangle_{\Fq} \\
&= \rk_{\Fq} \textstyle\sum_{j=1}^{\npr/g} \left\langle f\!\left(\alpha_j \right)_{a_i'} \right\rangle_{\Fqg} \\
&= \rk_{\Fq} \left\langle f\!\left(\alpha_{1} \right)_{a_i'}, \dots, f\!\left(\alpha_{\npr/g} \right)_{a_i'} \right\rangle_{\Fqg} \\
&= g \cdot \wtRWITHFIELDSIZE{q^g} \big[ f\!\left(\alpha_1\right)_{a_i'}, \dots, f\!\left(\alpha_{\npr/g}\right)_{a_i'} \big].
\end{align*}
where $\mathrm{(ii)}$ follows by the linearity of $f$ over $\Fqg$ and $\mathrm{(i)}$ holds since,
for $c \in \Fqm$ and $a \in \Fqm^\ast$, we have
\begin{align*}
\tilde{f}\!\left( c \right)_{a} &= \textstyle\sum_{i=0}^{n/g-1} f_i \sigma^{gi}(c) \prod_{j=0}^{ig-1} \sigma^j(a) \\
&= \textstyle\sum_{i=0}^{n/g-1} f_i \left(\sigma^{g}\right)^i(c) \prod_{j=0}^{i-1} \left(\sigma^g\right)^j\left( \prod_{\mu=0}^{g-1} \sigma^\mu(a)\right)\\
&= f(c)_{\prod_{\mu=0}^{g-1} \sigma^\mu(a)}.
\end{align*}
In summary, we have
\begin{align*}
\SparseSkewPolys_g^{(\a,\betaVec)}(\sigma,\tau) = \textstyle
    \left\{ \sum_{i=0}^d f_i x^{gi}  \, : \, \sum_{i=0}^d f_i x^{i} \in \SparseSkewPolys_1^{(\a',\alphaVec)}(\sigma^g,\tfrac{\tau}{g})\right\}.
\end{align*}
Thus, the lower bound on $\big|\SparseSkewPolys_g^{(\a,\betaVec)}(q,\tau)\big|$
follows by \cref{eq:cardinality_dense_polynomials} and the fact that $\ell \mid \tau$.
\end{proof}

\begin{corollary}\label{cor:list_size_general}
Assume the same setting as in \cref{thm:lower_bound_number_of_sparse_polys} and choose $k > n-2\tau$.
Let $\Code$ be the linearized RS code of parameters $[n,k,d]_{q^m}$, block size $\eta$, number of blocks $\ell$, and evaluation point vectors $\a$ and $\betaVec$. Then
\begin{align}
  \ListSize(\Code, \tau) \geq q^{m+\frac{\tau}{g} \left(\npr-m - \frac{\tau}{\ell}\right)} \gamma_{q^g}^{-\ell}. \label{eq:number_of_codewords_in_ball}
\end{align}
\end{corollary}

\begin{proof}
  Combine \cref{thm:lower_bound_number_of_sparse_polys} with \cref{lem:listsize_lemma}, and note that the polynomials in $\SparseSkewPolys_g^{(\a,\betaVec)}(q,\tau)$ may only have up to $\tfrac{2\tau}{g}$ non-zero monomials of degree $k = n - 2\tau$ or higher.
\end{proof}

\subsection{Families of LRS Codes with Exponential List Size Above the Unique Decoding Radius}

The following construction gives families of LRS codes with exponential list size directly above the unique decoding radius. For $\ell=1$ (rank metric), the construction coincides with the example families constructed in \cite[Section~IV]{raviv2016some}.

\begin{construction}\label{constr:family_Cg}
Fix $\ell, \an,\atau \in \ZZ_{>0}$ with $\ell < q$, and $\ell \mid \an$, and $\an > \max\{\atau^2,2\atau\}$. %
Define a family of LRS codes $\{\Code_g\}_{g \in \ZZ_{>0}, \, \ell \mid g}$ by choosing $\Code_g$ to have the code parameters
\begin{itemize}
\item $n= \an g$
\item $k=n-2\atau g+1$
\item $m=\eta = \tfrac{\an g}{\ell}$
\end{itemize}
and $\a \in (\Fqm^\ast)^\ell$ and $\betaVec \in \Fqm^\eta$ chosen in any way to satisfy \cref{thm:lower_bound_number_of_sparse_polys}.
\end{construction}

The following theorem shows that the code families in \cref{constr:family_Cg} an exponential list size directly above half the minimum distance $\lfloor \tfrac{n-k}{2}\rfloor$.

\begin{theorem}\label{thm:construction_rate_and_list_size}
Let $\ell,\an,\atau, %
\{\Code_g\}_{g \in \ZZ_{>0}, \, \ell \mid g}$
be chosen as in \cref{constr:family_Cg}. Then,
\begin{enumerate}
\item For $g \to \infty$, the rate of $\Code_g$ converges to
\begin{align*}
R(\Code_g) \to 1-\tfrac{2\atau}{\an} > 1 -\tfrac{2}{\sqrt{\an}}.
\end{align*}
\item For $g \to \infty$, we have
\begin{align*}
  \ListSize(\Code_g, \lfloor \tfrac{n-k}{2}\rfloor +1) \in \Omega\left(q^{c n}\right),
\end{align*}
where $c$ is a positive constant that depends only on $\ell,\an,\atau$.
\end{enumerate}
\end{theorem}

\begin{proof}
Ad 1): We have
\begin{align*}
R(\Code_g) = \tfrac{k}{n} = 1-\tfrac{2\atau}{\an}-\tfrac{1}{\an g} \to 1-\tfrac{2\atau}{\an}.
\end{align*}
The inequality follows by $\an \geq \atau^2+1$.

Ad 2): First note that $k > n-2\tau$ with $\tau_g := \lfloor \frac{n-k}{2}\rfloor +1 = Dg$. Since by the choice of $n,m,\eta,k,\tau,\alphaVec,\a$, all conditions of  \cref{cor:list_size_general} are fulfilled, then
\begin{align*}
\ListSize(\Code_g, \tau_g)
&\geq q^{m+\frac{\tau}{g} \left(\npr-m - \frac{\tau}{\ell}\right)} \gamma_{q^g}^{-\ell} \\
&= q^{\frac{g}{\ell} (\an-\atau^2)}\gamma_{q^g}^{-\ell}.
\end{align*}
The claim follows due to $\gamma_{q^g}^{-\ell} \to 1$, $\ell$ being a constant, and $\an-\atau^2>0$.
\end{proof}

We can choose a family as in \cref{constr:family_Cg} with a resulting rate $R$ arbitrarily close to $1$ by setting $\atau = 1$ and $\an$ any positive multiple of $\ell$ greater than $2$.
For $\ell>2$, the least possible rate we can achieve is obtained by choosing $\an = \ell$ and $D = \lfloor \sqrt{\an - 1} \rfloor$, resulting in $R \approx 1 - \frac 2 {\sqrt \ell}$.
\cref{tab:minimal_achievable_rates} shows the exact values of the smallest achievable rates for $\ell\leq 20$, as well as the corresponding constants $\an$ an $\atau$.

\begin{table}[ht!]
\caption{Minimal achievable rate by the code families in \cref{constr:family_Cg}, depending on the number of blocks $\ell$, for small values of $\ell$.}
\label{tab:minimal_achievable_rates}
\centering
{
\renewcommand{\arraystretch}{1.2}
\setlength{\tabcolsep}{5pt}
\begin{tabular}{c|c|c|c}
$\ell$ & $R$ & $\an$ & $\atau$ \\
\hline
\hline
$1$ & $0.200000$ & $5$ & $2$ \\
$2$ & $0.333333$ & $6$ & $2$ \\
$3$ & $0.333333$ & $3$ & $1$ \\
$4$ & $0.500000$ & $4$ & $1$ \\
$5$ & $0.200000$ & $5$ & $2$ \\
\hline 
$6$ & $0.333333$ & $6$ & $2$ \\
$7$ & $0.428571$ & $7$ & $2$ \\
$8$ & $0.500000$ & $8$ & $2$ \\
$9$ & $0.555556$ & $9$ & $2$ \\
$10$ & $0.400000$ & $10$ & $3$
\end{tabular}
\quad \quad
\begin{tabular}{c|c|c|c}
$\ell$ & $R$ & $\an$ & $\atau$ \\
\hline
\hline
$11$ & $0.454545$ & $11$ & $3$ \\
$12$ & $0.500000$ & $12$ & $3$ \\
$13$ & $0.538462$ & $13$ & $3$ \\
$14$ & $0.571429$ & $14$ & $3$ \\
$15$ & $0.600000$ & $15$ & $3$ \\
\hline 
$16$ & $0.625000$ & $16$ & $3$ \\
$17$ & $0.529412$ & $17$ & $4$ \\
$18$ & $0.555556$ & $18$ & $4$ \\
$19$ & $0.578947$ & $19$ & $4$ \\
$20$ & $0.600000$ & $20$ & $4$
\end{tabular}
}
\end{table}

\bibliographystyle{IEEEtran}

\bibliography{main}

\end{document}